\documentclass{article}
\usepackage{bm}
\usepackage[T2A]{fontenc}
\usepackage[utf8]{inputenc}
\usepackage[english]{babel}
\usepackage{amsmath,amssymb,amsfonts}
\usepackage[a5paper,nohead,headsep=0.6cm,left=1.5cm,
right=1cm,top=2cm,bottom=2cm,twoside]{geometry}
\usepackage{graphicx}

\newtheorem{theorem}{Theorem}

\begin{document}

\begin{center}
	\title{}{\bf AN ADDITIONAL POSSIBILITIES OF THE STANDARD METHOD
 OF INVERTING THE RADON TRANSFORM}
	
	\author{}{D.S. Anikonov}, {S.G. Kazantsev}, {D.S. Konovalova}
	
Sobolev Institute of Mathematics,  Novosibirsk, Russia

e-mail:anik@math.nsc.ru, kazan@math.nsc.ru, dsk@math.nsc.ru

\end{center}

 \begin{quote}
	\noindent{\bf Abstract. } The problem of inverting the Radon integral transform in
    finite-dimensional Euclidean space is considered. The relevance
    of this topic for probing issues is    indicated. It is noted that for the latter direction, it is
    natural to consider the integrand  as discontinuous function.
    However, the available inversion formulas are only proven for
    differentiable functions. Therefore, the question  of obtaining
    formulas for discontinuous functions arises. It is set that the
    required results can be obtained by some modification of available
    proofs for smooth functions.
	
\end{quote}

\section*{ Introduction and Statements}
The Radon transform is defined as integrals of a function over
hyperplanes in $n$-dimensional Euclidean space. This transform has
been used in many areas of mathematics. For example,  in the
theory of hyperbolic differential equations, it was possible to
reduce the dimension  of space to two. In this simple case, it was
possible to use explicit representations for solutions  of Cauchy
problems. However, such representations were obtained for the
images of the Radon transform.   Therefore, the next stage of the
work was the question of inverting the transform. These
mathematical   aspects have been investigated, for example, in
\cite{Courant:1962} -- \cite{Natterer:1986}. It is important to keep
in mind that the obtained inversion formulas were proven only for
smooth functions. Another application of the Radon transform
relates to the theory of probing unknown media with various
physical signals. In particular,   the mathematical part of X-ray
tomography theory is reduced to the inverse Radon transform for
$n=2$ and $n=3$. Furthermore, the Radon transform is widely used
in seismic exploration, geophysical mapping, defectoscopy, and
environmental monitoring. It is emphasized that for probing, it is
natural to consider the sought characteristics as discontinuous
functions. This is the way to describe inhomogeneous media
consisting of different materials. In this paper we prove formulas
for the inversion of the Radon transform for discontinuous
functions.

Below we present a short review of works closed to us. In the
articles \cite{K:2020} -- \cite{OZCC:2024}, numerical methods are
used to study practically important visualization problems,
inverse problems of seismic exploration in layered media of
various dimensions, and others. It is natural to look for a
solution to such problems in the class of piecewise continuous
functions, which is what the authors do. These studies also use
the discrete Radon transform and its inversion. The question of
the accuracy of such discretization is studied.

The present work is a continuation of results published in
\cite{AKK:2023} -- \cite{AKo:2024}.
Mainly in these  works the surfaces of discontinuous integrand were
determined with respect to generalized Radon transforms. Also the according algorithms were presented. Unlike
the previous results, here we determine the values of integrand in
points of its continuity for classical Radon transform.

Generally speaking, theory of Radon transform is growing in many
directions. In particular, we note to interesting article
 devoted to vector
and tensor Radon transforms  \cite{L:2024,Ky:2025}.

In the present work  a slight modification of available
proofs allows  us to obtain explicit inversion formulas for
discontinuous integrand. There are grounds to believe that the
application of our formulas will noticeably expand the
possibilities of using corresponding algorithms.

\section{Definition and preliminaries}

We will use the following notations. $\mathbb{E}_n$ means  an
$n$-dimensional Euclidean space; ${\bf e}_{1}, \ldots,{\bf e}_{n}$
is the main orthonormal basis in $\mathbb{E}_n$; points ${\bf x},
{\bf y}, { \bm\omega}$  in the space $\mathbb{E}_n$ in the main
basis are represented as ${\bf x}=\left(x_{1}, \ldots,
x_{n}\right)$, ${\bf y}=\left(y_{1}, \ldots, y_{n}\right)$, ${\bm
    \omega}=\left(\omega_{1}, \ldots, \omega_{n}\right)$; $\Omega$
denotes the unit sphere in $\mathbb{E}_n$; ${\bm \omega}$ is an
element of the sphere in $\mathbb{E}_n$, $\Omega=\{{\bm \omega}:
{\bm \omega } \in \mathbb{E}_n,|{\bm \omega}|\}=1\}$; $Y({\bm
    \omega}, p)=\{{\bf y}: {\bf y}\cdot{\bm \omega}=p\}$, $ {\bf y}
\in \mathbb{E}_n$, ${\bm \omega}\in \Omega$, $-\infty < p <
\infty$ represents a hyperplane in $\mathbb{E}_n$; $\partial T$
denotes the boundary of the set  $T \subset\mathbb{E}_n$.

Let $G \subset\mathbb{E}_n$ be a bounded region, containing
pairwise disjoint regions $G_{i}$, $i=1, \ldots, N$. Denoting the
union of these subsets by $G_0=\cup_{i=1}^{N}G_{i}$, we require
that $\overline{G}_{0} = \overline{G}$ and  assume that the
boundary of the set $G_0$ has zero measure in the space
$\mathbb{E}_n$.

Define now  the class of functions $F$. Function $f({\bf y})$, ${\bf y} \in \mathbb{E}_n$, belongs to
$F$, if the following conditions are satisfied:
$|f({\bf x}) - f({\bf y})|
\leq \textit{const}\, |{\bf x}-{\bf y}|^{\alpha}$, ${\bf x}, {\bf y} \in G_i$, $0< \alpha \leq 1$ and $f({\bf y}) = 0$ for ${\bf y} \in \mathbb{E}_n \setminus G$, where
$\textit{const}$ is a positive number, $i=1, \ldots, N$.

Let $f({\bf y})$ be a function defined for ${\bf y} \in
\mathbb{E}_n$, which is integrable in $\mathbb{E}_n$. The Radon
transform is defined by the following equality
\begin{equation}
    [R f]({\bm \omega}, p)=\int\limits_{{\bm \omega} \cdot {\bf y}=p} f({\bf y}) d_{\bf y} \sigma,
    \label{Radon}
\end{equation}
where the right-hand side contains a surface integral over the
hyperplane $Y({\bm \omega}, p)$,
see \cite{John:1955, Helgason:1999,  Natterer:1986}.

Using Fubini's theorem for any ${\bm\omega} \in \Omega$, we have
\begin{equation}
    \int_{G} f({\bf y}) d{\bf y}=\int\limits_{-\infty}^{+\infty} \int\limits_{{\bf y} \cdot {\bm \omega}=p} f({\bf y}) d_{\bf y} \sigma dp=\int\limits_{-\infty}^{+\infty}[R f]({\bm \omega}, p) dp.
    \label{f2}
\end{equation}

In this work the following problem is set    and investigated.

{\bf Problem:} Given the function $[Rf]({\bm \omega }, p)$ for
${\bm \omega } \in \Omega$, $p \in (-\infty, +\infty)$,
find $f({\bf y})$, ${\bf y} \in \mathbb{E}_n \backslash \partial G_0$, where $f \in F$.

\section{Usefull formulas for the Radon transform }
In the following, we will assume that $f \in F$.

Similarly to identities (\ref{f2}), for any continuous and local
integrable  function $\psi(s)$, $s \in (0, \infty)$, we write the
equalities
\begin{align}\nonumber
    &\int_{G} f({\bf y}) \psi  \Big( | ({\bf y}-{\bf x}) \cdot {\bm \omega }| \Big) d{\bf y}
    \\
		&=\int\limits_{-\infty}^{\infty} \int\limits_{({\bf y}-{\bf x})\cdot {\bm \omega }=p}
    f({\bf y}) \psi(|p|) d_{\bf y} \sigma d p= \int\limits_{-\infty}^{+\infty}
    \psi(|p|)[R f]({\bm \omega}, p + {\bf x}\cdot {\bm \omega }) dp.
    \label{f3}
\end{align}

We will need also the following identities
\begin{align}
    \int\limits_{\Omega}|{\bm \omega} \cdot {\bf x}| d{\bm \omega}& =\beta_{1,n}|{\bf x}|,  \quad {\bf x} \in \mathbb{E}_n,
    \label{f4}
    \\
    (\Delta_{\bf x})^{\frac{n-1}{2}}|{\bf y}-{\bf x}|& =\beta_{2, n}
    |{\bf y}-{\bf x}|^{2-n},  \quad n \geqslant 3,\, n=2m+1,
    \label{f5}
    \\
    \Delta_{\bf x} \int\limits_{G} \frac{f({\bf y})}
    {|{\bf y}-{\bf x}|^{n-2}} d{\bf y}& =\beta_{3,n} f({\bf x}), \  {\bf x} \in \mathbb{E}_n \backslash\partial G_{0},  \quad  n \geqslant 3,
    \label{f6}
\end{align}
where
$$
\beta_{1, n}=\frac{2 \pi^{\frac{n-1}{2} }}{\Gamma\left(\frac{n+1}{2}\right)}, \
\beta_{2, n} =(-1)^{\frac{n-1}{2}}\frac{2^{n} \Gamma(\frac{3}{2}) \Gamma\left(\frac{n+1}{2}\right) \Gamma\left(\frac{n}{2}\right)}{ \pi(2-n)}, \
\beta_{3, n} =\frac{2{\pi^{\frac{n}{2}}(2-n) } }{\Gamma\left(\frac{n}{2}\right)}.
$$
This  identities provided, for example, in \cite{John:1955}, formulas (1.6), (1.9a) and (1.8)--(1.9).

\subsection{Inversion of the Radon transform for odd $n=2m+1$}
\begin{theorem}
    \label{theorem_1}
    If $f \in F$ and $n=2m+1$, $m=1, \ldots$, then for all ${\bf x} \in \mathbb{E}_n \backslash \partial G_{0}$ the following equality holds
    \begin{equation}
        \beta_n f({\bf x}) = \left(\Delta_{{\bf x}}\right)^{\frac{n+1}{2}} \int\limits_{\Omega} \int\limits_{-\infty}^{+\infty} |p|[Rf]({\bm \omega}, p + {\bf x} \cdot {\bm \omega}) dp d{\bm \omega}, \quad \beta_{n} \neq 0.
        \label{f7}
    \end{equation}
\end{theorem}

{\bf Proof.} Consider the function
    \begin{equation}
        U_{1}({\bf x})=\int\limits_{\Omega} \int\limits_{G} f({\bf y})|({\bf y-}{\bf x}) \cdot {\bm \omega}| d{\bf y} d{\bm \omega}, \quad {\bf x} \in  \mathbb{E}_n \backslash \partial G_{0}.
        \label{f8}
    \end{equation}
    Using the identities (\ref{f3}) for the function $\psi(s)=|s|$, we
    obtain
    \begin{equation}
        U_{1}({\bf x}) = \int\limits_{\Omega} \int\limits_{-\infty}^{\infty} |p|[R f]({\bm \omega}, p+{\bf x}\cdot{\bm \omega}) dp d{\bm \omega}.
        \label{f9}
    \end{equation}
    Now, let's transform $U_{1}({\bf x})$ in a different way. Changing the
    order of integration in the right-hand side of equality (\ref{f8}) and
    considering the identity (\ref{f4}), we can write
    $$
    U_{1}({\bf x}) =
    \beta_{1,n} \int\limits_{G} f({\bf y})|{\bf y}-{\bf x}| d{\bf y}.
    $$
    Next, we apply the operator $(\Delta_{{\bf x}})^{\frac{n-1}{2}}$ to this equality, then using identity (\ref{f5}), we derive
    \begin{equation}
        \left(\Delta_{{\bf x}}\right)^{\frac{n-1}{2}} U_{1}({\bf x}) = \beta_{1, n}
        \beta_{2, n} \int\limits_{G} \frac{f({\bf y})}{|{\bf y}-{\bf x}|^{n-2}} d{\bf y}.
        \label{f10}
    \end{equation}
    Now, applying the operator $\Delta_{\bf x}$ to both sides of the
    latter equality and using property (\ref{f6}),  we obtain
    \begin{equation}
        \left(\Delta_{{\bf x}}\right)^{\frac{n+1}{2}} U_{1}({\bf x}) = \beta_{1, n}
        \beta_{2, n}  \beta_{3, n} f({\bf x}).
        \label{f10}
    \end{equation}
    Let $\beta_{n} = \beta_{1,n}  \beta_{2,n}  \beta_{3,n}$ in  (\ref{f10}) and using expression (\ref{f9})
    for function $U_{1}({\bf x})$, we  get equality (\ref{f7}). The theorem is proven.
$\square$

    It is worth clarifying that the presented proof is analogous to a segment of the proof of the
    inversion formula in \cite{John:1955}. In our notation, the formula in
    \cite{John:1955}  takes the form
    \begin{equation}
        2(2 \pi i)^{n-1} f({\bf x}) = \left(\Delta_{\bf x}\right)^{\frac{n-1}{2}} \int\limits_{\Omega} [Rf]({\bm \omega}, {\bm \omega} \cdot {\bf x}) d{\bm \omega}.
        \label{f12}
    \end{equation}

    Now let's compare formulas (\ref{f7}) and (\ref{f12}). We can see that equality (\ref{f12}) contains one
    less integral than  equality (\ref{f7}). The power of the Laplace operator in (\ref{f12}) is also smaller
    than in (\ref{f7}). In other words, finding the function $f({\bf x})$ using (\ref{f12}) is a simpler
    operation than a similar deduction using (\ref{f7}). However, equality (\ref{f12}) is
    proven for smooth functions, while equality (\ref{f7}) holds true for discontinuous functions as well.
    Thus, it can be said that (\ref{f7}) and
    (\ref{f12}) have their own merits and drawbacks.

\subsection {Inversion of the Radon transform for even $n=2m$}

Here we will use the following identities (see   formulas (1.7) and (1.9b) in \cite {John:1955}). Let $ {\bf x}, {\bf y}\in \mathbb{E}_n$, then
\begin{align}
    \int\limits_{\Omega} \ln |({\bf y-}{\bf x}) \cdot {\bm \omega}| d {\bm \omega }& =
    \gamma_{1, n}(\ln |{\bf y}-{\bf x}|+\gamma_{0, n}), \quad
    \  \quad \gamma_{0, n} \neq 0,
    \label{f13}
    \\
    \left(\Delta_{\bf x}\right)^{\frac{n-2}{2}} \ln |{\bf y}-{\bf x}|
    & = \gamma_{2, n}|{\bf y}-{\bf x}|^{2-n},
    \quad n \neq 2,
    \label{f14}
\end{align}
where
$$
\gamma_{1, n} = \frac{2\pi^{\frac{n}{2}}}{\Gamma\left(\frac{n}{2}\right)}, \
\gamma_{2, n} = (-1)^{\frac{n-2}{2}}
\frac{2^{n-2} \Gamma^{2}\left(\frac{n}{2}\right)}{2-n}.
$$

For $n=2$ we have
\begin{equation}
    \Delta_{\bf x} \int\limits_{G} f({\bf y}) \ln |{\bf y}-{\bf x}| d{\bf y} = 2\pi f({\bf x}), \quad {\bf x} \in \mathbb{E}_2 \backslash \partial G_{0}.
    \label{f15}
\end{equation}

\begin{theorem}
    \label{theorem_04}
    Let function  $f \in F$ and  $n=2m$, $m=1, \ldots$,   then for all ${\bf x} \in \mathbb{E}_n \backslash \partial G_{0}$ the next equality holds
    \begin{equation}
        \gamma_{n} f({\bf x}) = (\Delta_{\bf x})^{\frac{n}{2}} \int\limits_{\Omega} \int\limits_{-\infty}^{+\infty} \ln |p|[R f]({\bm \omega}, p+{\bf x}\cdot {\bm \omega})d p d {\bm \omega}, \quad \gamma_{n} \neq 0.
        \label{f16}
    \end{equation}
\end{theorem}

{\bf Proof}
 Let us define the function
    \begin{equation}
        U_{2}({\bf x}) = \int\limits_{G} \int\limits_{\Omega} f({\bf y}) \ln |({\bf y}-{\bf x}) \cdot {\bm \omega})| d {\bm \omega} d{\bf y}, \quad {\bf x} \in \mathbb{E}_n \backslash \partial G_{0}.
        \label{f17}
    \end{equation}
    We change the order of integration in the right-hand side of the
    last equality and use identities (\ref{f3}) with $\psi(s)= \ln \,
    s .$ Then we obtain
    \begin{align}
        U_{2}({\bf x}) = \int\limits_{\Omega} \int\limits_{G} f({\bf y}) \ln |({\bf y-}{\bf x})
        \cdot {\bm \omega}| d{\bf y} d {\bm \omega} = \int\limits_{\Omega} \int\limits_{-\infty}^{\infty}
        \ln |p|[R f]({\bm \omega}, p+{\bf x} \cdot {\bm \omega}) dp d {\bm \omega}.
        \label{f18}
    \end{align}
    Also we can  give another representation  for the same function
    $U_2({\bf x})$. For this we  transform the inner integral in
    equality  (\ref{f17}) by using identity (\ref{f13}), we get
    \begin{equation}
        U_{2}({\bf x})=\int\limits_{G} f({\bf y})\left(\gamma_{1, n} \ln |{\bf y}-{\bf x}|+\gamma_{0, n}\right) d {\bf y}.
        \label{f19}
    \end{equation}
    Next, the cases $n=2 $ and $  n>2$ are considered separately. For  $n>2$ we apply the operator $(\Delta_{\bf x})^{\frac{n-2}{2}}$
    to equality (\ref{f19}) and using identity (\ref{f14}), we obtain
    \begin{equation}
        \left(\Delta_{\bf x}\right)^{\frac{n-2}{2}} U_{2}({\bf x})=\gamma_{1, n} \gamma_{2, n} \int\limits_{G} \frac{f({\bf y})}{ |{\bf y}-{\bf x}|^{n-2}} d{\bf y}.
        \label{8}
    \end{equation}
    Applying now the operator $\Delta_{\bf x}$ to both sides of the above equality and using property (\ref{f6}), we obtain
    \begin{equation}
        \left(\Delta_{\bf x}\right)^{\frac{n}{2}} U_{2}({\bf x})=\gamma_{1, n}  \gamma_{2, n}
        \gamma_{3, n} f({\bf x}),  \ \gamma_{3, n}=\beta_{3, n}.
        \label{f21}
    \end{equation}
    Let's denote  $\gamma_{n}=\gamma_{1, n}  \gamma_{2, n} \gamma_{3, n}$ and compare equalities
    (\ref{f18}) and (\ref{f21}). As a result, we have the identity (\ref{f16}).

    Now, let  $n=2$.  Applying the operator $\Delta_{{\bf x}}$ to (\ref{f19}) 
		and using identity (\ref{f15}), we write
    \begin{equation}
        \Delta_{\bf x}U_{2}({\bf x})= 2 \pi\gamma_{1,2}  f({\bf x}).
        \label{f22}
    \end{equation}
    Let's denote $ 2 \pi\gamma_{1,2}=\gamma_{2}$ and compare (\ref{f18}) and (\ref{f22}). As a result,
    we obtain the equality (\ref{f16}) for the particular case $n=2$. The theorem is proven.
$\square$
\begin{equation}
    (2 \pi i)^{n} f({\bf x})=(\Delta_{\bf x} )^{\frac{n-2}{2}}
    \int\limits_{\Omega} \int\limits_{-\infty}^{\infty} \ln |p-{\bf x} \cdot {\bm \omega}| \frac{\partial^{2}}{\partial p^{2}}[R f]({\bm \omega}, p)dp d{\bm \omega}.
    \label{f23}
\end{equation}

Let's compare the inversion formulas (\ref{f16}) and (\ref{f23}).
Each of them involves the same order of integration and
differentiation. However, to evaluate  the equality (\ref{f23}),
it is necessary to assume that the function $f({\bf y})$ has
continuous and bounded derivatives up to second order. But for the
formula (\ref{f16}) it is sufficient to require that $f({\bf y})$
is piecewise continuous in the Holder sense. Hence, we can
conclude that formula (\ref{f16}) is more valuable for
applications than (\ref{f23}).
\section{Conclusion}

To the accepted restrictions in this article we add the condition
that  the $n$-dimensional Lebesgue measure of the set $\partial G_0$ is equal to zero. This condition is used in mathematical physics very rarely, but here it allows us to assert that formulas
(\ref{f7}) and (\ref{f16}) are true for almost all ${\bf x}\in
E_n$. In general, identities (\ref{f7}) and (\ref{f16}) can be
considered as the theoretical basis for the numerical algorithm for finding $f({\bf x})$. By using the discrete set for the problem, the values of $f({\bf x})$ can be approximately determined. Also the set $\partial G_0$ is approximately founded as a set of points at which the function $f({\bf x})$ has an anomalously large jumps.


\begin{thebibliography}{10}
\setlength{\parsep}{0pt}\setlength{\itemsep}{3pt}

\bibitem{Courant:1962}
Courant~R., \emph{Partial differential equations}, Interscience, France,
1962.

\bibitem{John:1955}
John~F.,
\emph{Plane waves and spherical means}, New York: Intersci. Publ., 1955.


\bibitem{Helgason:1999}
Helgason~S., Helgason, \emph{ Radon Transform}, Birkhauser, Basel, Switzerland, 2nd ed.,
1999.


\bibitem{Natterer:1986}
Natterer~F.,
\emph{The Mathematics of Computerized Tomography}, Wiley, Hoboken, NJ, USA, 1986.

\bibitem{K:2020}
Katsevich~A.,
Analysis of reconstruction from discrete
Radon transform data in $R^3$ when the function has jump
discontinuities, 
\emph {SIAM Journal on Mathematical Analysis}
\textbf{52}(4)  (2020),
3990--4021.


\bibitem{OZCC:2024}
Olugboji~T., Zhang~Z., Carr~S., Cetin~C.,
On the detection of upper mantle discontinuities with
radon-transformed receiver functions,
\emph{Geophys. J. Int.} 
\textbf{236} (2024),
748--763.


\bibitem{AKK:2023}
Anikonov~D.\,S.,   Kazantsev~S.\,G., Konovalova~D.\,S.,
A uniqueness result for the inverse problem of identifying boundaries from weighted Radon transform,
\emph{Journal of Inverse
and Ill-Posed Problems} \textbf{31} (2023),
959--965.

\bibitem{AK:2024}
Anikonov~D.\,S.,  Konovalova~D.\,S.,
Radon transform inversion formula in the class of discontinuous functions,
\emph{J. Appl. Ind. Math.}
\textbf{8} (2023), 379--383.

\bibitem{AKo:2024}
Anikonov~D.\,S.,  Konovalova~D.\,S., Inversion problem for Radon
transforms defined on pseudoconvex sets, 
\emph{ Dokl. Math.}
\textbf{109} (2024),  175--178.

\bibitem{L:2024}
Louis~A.\,K., A unified approach to inversion formulae for vector and tensor ray and Radon transforms and the Natterer inequality,
\emph {Inverse Problems}
\textbf{40}(8) (2024), 085007.


\bibitem{Ky:2025}
Kunyansky~L.,  McDugald~E.,  Shearer~B.,
Weighted Radon
transforms of vector fields, with applications to
magnetoacoustoelectric tomography, 
\emph {Inverse Problems},
\textbf{39} (2023), 065014.
\end{thebibliography}
\end{document}